\documentclass[a4paper]{oasics-v2018}
\nolinenumbers

\usepackage[utf8]{inputenc}
\usepackage{microtype}
\usepackage{amsmath}
\usepackage[most]{tcolorbox}
\usepackage{cleveref}
\usepackage{bbm}

\DeclareMathOperator{\polylog}{\mathrm{polylog}}
\DeclareMathOperator{\poly}{\mathrm{poly}}
\newcommand{\id}[1]{\mathbbm{1}_{#1}}
\newcommand{\prob}[1]{\mathrm{Pr}\left[{#1}\right]}

\title{Improved Algorithms for Edge Colouring in the W-Streaming Model}
\author{Moses Charikar and Paul Liu}{Stanford University, USA}{\{moses, paul.liu\}@stanford.edu}{}{}

\authorrunning{M. Charikar and P. Liu} 

\Copyright{M. Charikar and P. Liu}

\subjclass{
Theory of computation $\rightarrow$ Streaming Models; Approximation algorithms analysis.
Mathematics of computing $\rightarrow$ Graph coloring.
} 

\ArticleNo{00}

\begin{document}

\maketitle

\begin{abstract}
In the W-streaming model, an algorithm is given $O(n \polylog n)$ space and must process a large graph of up to $O(n^2)$ edges. In this short note we give two algorithms for edge colouring under the W-streaming model. For edge colouring in W-streaming, a colour for every edge must be determined by the time all the edges are streamed. Our first algorithm uses $\Delta + o(\Delta)$ colours in $O(n \log^2 n)$ space when the edges arrive according to a uniformly random permutation. The second  algorithm uses $(1 + o(1))\Delta^2 / s$ colours in $\tilde{O}(n s)$ space when edges arrival adversarially.
\end{abstract}

The problem of edge colouring asks for an assignment of colours to the edges of a graph such that no two incident edges have the same colour and the total number of colours used is minimized. In this short note, we study the problem of edge colouring in the W-streaming model.

In the W-streaming model, input is given in a streaming fashion, and an output stream is written as the input is processed. An algorithm can choose to stream over the input many times, although for our purposes we focus on one-pass algorithms. With respect to edge colouring in W-streaming, an algorithm is given $O(n \polylog n)$ space and must process a large graph of $n$ vertices, up to $O(n^2)$ edges, and maximum degree $\Delta$. During the course of the stream, each edge must be coloured by the algorithm, and their colours announced by the time the stream completes. Several classical graph problems have been studied in the W-streaming model, including connectivity, minimum spanning tree, euler tours, and in particular edge colouring~\cite{BDHKS19, DEMR10, DFR09, GSS17}. Since only $O(n\polylog(n))$ space is allowed, most edges have their colours announced close to the time they are streamed in. In contrast to the online model, edge colours do not have to be announced the instant they are seen. Instead, we may use our $O(n \polylog n)$ space to ``buffer'' the announcement of the edge colours.

In the online model, a lower bound of Bar-Noy et al.~\cite{BMN92} shows that no algorithm can do better than a $2\Delta$ colouring when $\Delta = O(\log n)$ ($2\Delta$ is achieved by the greedy algorithm). On the algorithmic side, several $O(\Delta)$ colouring algorithms are known, with approximation factors depending on whether the input is \emph{vertex-arrival} or \emph{edge-arrival}. In the vertex arrival model, all edges out-going from a vertex are streamed consecutively in a batch. In the edge-arrival model, there are no constraints on the order of the edges. For \emph{adversarial orders} in the vertex-arrival model, the best known result is achieved by Cohen et al.~\cite{CPW19} using $(1+o(1))\Delta$ colours. For \emph{random orders} in the edge-arrival model, an algorithm  using $(1+o(1))\Delta$ colours by Bhattacharya et al.~\cite{BGW20} was simultaneously discovered at the same time as this paper, surpassing the previous state-of-the-art of $1.26\Delta$ by Bahmani et al.~\cite{BMM10}. However, the online model allows for $O(\poly(n))$ memory, so the methods taken by previous works do not apply in W-streaming. In particular, even the classic greedy algorithm requires $O(n\Delta)$ space. In the offline model, algorithmic versions of Vizing's theorem~\cite{D00, MG92} gives a polynomial time $\Delta+1$ colouring algorithm for edge colouring. Any graph of maximum degree $\Delta$ requires at least $\Delta$ colours, so distinguishing between $\Delta$ and $\Delta+1$ known to be NP-Hard.

The current state-of-the-art in W-streaming is by Behnezhad et al.~\cite{BDHKS19}, who achieve a $2e\Delta$-colouring algorithm for the random arrival case and a $O(\Delta^2)$ algorithm for the adversarial arrival case. In this short note, we give two algorithms for edge colouring under the W-streaming model that improves upon Behnezhad et al.~\cite{BDHKS19}. The first is a $\Delta + o(\Delta)$ algorithm for the random arrival case using $O(n \log n)$ space. The second is a $O(\Delta^2 / s)$ algorithm for the adversarial arrival case using $\tilde{O}(n s)$ space.

\section{A simple algorithm for edge colouring under random arrivals}
When the edges are streamed in randomly, consider the following simple algorithm:
\begin{tcolorbox}[title={Algorithm 1: A simple algorithm for random arrivals}]
\label{alg:simple-random}
Break the edges into $C$ sized chunks where $C = \alpha^2 n$. For each chunk, use any offline $\Delta+1$ coloring algorithm. For each chunk use a different palette of colours.
\end{tcolorbox}

\begin{theorem}
When $\alpha = \Omega(\log n)$, Algorithm~\ref{alg:simple-random} colours any input graph with $\Delta + O(\Delta/\sqrt{\alpha})$ colours with high probability. For a $\Delta+o(\Delta)$ colouring, $O(n \log^2 n)$ space is sufficient.
\end{theorem}
\begin{proof}
Let $G_i$ be the subgraph induced by chunk $i$, and let $d(u)$ and $d_i(u)$ be the degree of $u$ in $G$ and $G_i$ respectively. For the ease of exposition, assume that $C$ divides $E$ exactly. The total number of chunks is $N = \frac{E}{C}$. Fix any $i \in \{1, \ldots, N\}$. Since the edges are randomly ordered, the expected degree of $u$ in $G_i$ is $\hat{d_u} = \frac{d(u)}{N}$ in expectation. For each of the $C$ edges in our chunk, let $\id{u,j}$ be an indicator variable for the event that the $j$-th edge is incident on node $u$. Thus, we have $d_i(u) = \sum_{j=1}^{C}  \id{u,j}$. Since the edges are ordered according to a random permutation, $d_i(u)$ follows a hypergeometric distribution and the $\id{u,j}$'s are negatively dependent. 

Applying Chernoff's bound for negatively dependent random variables (\Cref{thm:chernoff}), we have $\prob{d_i(u) \geq (1+\delta_u) \hat{d_u}} \leq \exp\left(-\min\{\delta_u, \delta_u^2\} \hat{d_u}/3\right)$. Choose $\delta_u = \frac{\Delta}{N\hat{d_u}\sqrt{\alpha}}$. Then for $\delta_u \leq 1$, we have:
\begin{align*}
\prob{d_i(u) \geq (1+\delta_u) \hat{d_u}} & \leq \exp\left(-\frac{\Delta^2}{3 N^2\hat{d_u} \alpha}\right) & \\
& \leq \exp\left(-\frac{\Delta}{3 N \alpha}\right) & \text{($d_u \leq \Delta / N$)} \\
& = \exp\left(-\frac{\alpha n \Delta}{3 E}\right) & \text{($N = \frac{E}{\alpha^2 n}$)}  \\
& \leq \exp\left(-\alpha/3\right) = 1/\poly(n) & \text{($E \leq n \Delta, \alpha = \Omega(\log n)$)} 
\end{align*}
where we may adjust the degree of the $\poly(n)$ as high as we like (by adjusting $\alpha$). The case for $\delta_u > 1$ achieves similar $1/\poly(n)$ bounds.

Thus for every $u$ in $G_i$, we have $d_i(u) \leq (1+\delta) \hat{d_u}$ with probability $1 - 1/\poly(n)$. Let $\id{i}$ be the indicator variable for this event happening for $G_i$. By a union bound, we have $$\prob{\cap_{i=1}^N \id{i}} \geq  1 - \sum_{i=1}^N (1 - \prob{\id{i}}) \geq 1 - 1/\poly(n).$$

Let $u$ be a node of maximum degree $\Delta$. The analysis above means we colour every chunk with at most $(1 + \delta_u) \hat{d_u} + 1$ colours.
Thus the total number of colours used is at most $N((1+\delta_u)\hat{d_u} + 1) = \Delta + O(\Delta/\sqrt{\alpha})$ with probability at least $1 - 1/\poly(n)$.
\end{proof}

\section{A simple algorithm for edge colouring under adversarial arrivals}
When the edges are streamed in adversarially, consider the following simple algorithm:
\begin{tcolorbox}[title={Algorithm 2: A simple algorithm for adversarial arrivals}]
\renewcommand{\thempfootnote}{\arabic{mpfootnote}}
\label{alg:simple-adversarial}
For each node $u$, generate $s > 36 \log n$ random bits. Partition the edge set $E$ into bipartite graphs $B_1, B_2, \ldots, B_s$, by the following procedure:
\begin{itemize}
\item Given $e=(u, v)$ let $\mathcal{D}$ be the indices where the random bits of $u$ and $v$ differ.
\item Choose an index $i\in \mathcal{D}$ uniformly at random and assign $(u,v)$ to $B_i$.
\end{itemize}
For each node $u \in B_i$, store a counter $C^{(i)}_u$ initially set to 0. Upon streaming in an edge $(u, v) \in B_i$,\footnote{Where $u$ lies in the left partition and $v$ lies in the right partition of its bipartite graph.} output the colour $(i, C^{(i)}_u, C^{(i)}_v)$. Increment both $C^{(i)}_u$ and $C^{(i)}_v$ by 1.\footnote{The version of the algorithm we use on each bipartite graph $B_i$ is the same as the one used in~\cite{BDHKS19}.}
\end{tcolorbox}

\begin{theorem}
Algorithm~\ref{alg:simple-adversarial} colours any input graph with $(1 + o(1))\Delta^2 / s$ colours in $\tilde{O}(ns)$ space.
\end{theorem}
\begin{proof}
First, we show that the graphs $B_i$ produced in Algorithm~\ref{alg:simple-adversarial} are indeed bipartite. By construction, every edge $(u,v)$ in graph $B_i$ connects a node whose $i$-th bit is 0 with a node whose $i$-th bit is 1. Thus the left partition of $B_i$ is precisely the nodes with their $i$-th bit equal to 0 and the right partition is precisely the nodes with their $i$-th bit equal to 1.

Next, we show that the maximum degree of each $B_i$ is concentrated about $\Delta / s$. For any two vertices $u$ and $v$, the expected number of differing random bits is $s/2$. By a Chernoff bound, 
$$\prob{\text{\# differing bits between $u$ and $v$} \leq s/4} \leq \exp(-s/12) \leq 1/n^3.$$
Thus by a union bound, every pair of nodes will have at least $s/4$ differing bits. By symmetry, any index is equally likely to differ between any pair of nodes. For each edge, we choose a random index out of all the differing indices to assign the edge to. Thus the maximum degree of each $B_i$ is $\Delta / s$ in expectation. Furthermore, the expectation is tightly concentrated around $\Delta / s$ when $\Delta/s = \Omega(\log n)$ (when $\Delta / s = o(\log n)$, $O(ns \polylog n)$ is enough space to store the entire graph).

Finally, we prove the correctness and space guarantees of the algorithm. In the colours we assign, the first index guarantees that each bipartite graph uses a distinct palette of colours. Now fix a particular bipartite graph $B_i$. We first argue correctness. Let $u$ be a node of the left partition. Since a counter is incremented whenever it's used, it's clear that no edge incident on node $u$ will ever reuse a value for the second coordinate of the colour. The right partition follows a similar argument for the third coordinate of the colour. Note the bipartite structure of the graph is crucial here, as it allows us to distinguish between the left node and right node of every edge. 

Since we increment by 1 on the addition of every edge incident to a node, the counter values can only be incremented as high as the maximum degree of the graph. As the maximum degree of each $B_i$ is $\Delta/s$ in expectation, the counter values are at most $\Delta/s$ in expecation (and tightly concentrated when $\Delta$ is large). Summing over all $s$ bipartite graphs, the total number of possible colours is tightly concentrated around $s(\Delta/s)^2 = \Delta^2/s$.

Finally, the algorithm uses at most $O(ns \log (\Delta / s))$ space, as each vertex uses $s$ counters, and each counter is at most $O(\Delta / s)$.
\end{proof}

\paragraph*{Worst case examples.} The bound of $O(\Delta^2/s)$ can be achieved in the worst case for Algorithm~\ref{alg:simple-adversarial} given access to the randomness of the algorithm. Let an $t$-star be the graph of $t$ edges incident to node $u$, and let $s$ be the space parameter $s$ in Algorithm~\ref{alg:simple-adversarial}. By streaming in $\poly(\Delta s)$ copies of a $ts$-star, we guarantee the existence of $t$-stars in each of the bipartite subgraphs $B_1, B_2, \ldots, B_s$ for any $t$. Now any colour $(i, j, k)$ may be created by connecting the centres of a $j$-star and a $k$-star in graph $B_i$. Note that we need to create $\poly(s)$ copies of the $j$-star and $k$-star to connect, so that the edge connecting the centres of the stars falls into $B_i$. Creating all possible combinations of $(i,j,k)$, we obtain a graph using $O(\Delta^2/s)$ colours in $\poly(\Delta s)$ nodes with maximum degree $\Delta$. 

\bibliography{main}
\bibliographystyle{plain}

\appendix

\section{Appendix}

\subsection{Auxiliary Results}
\begin{theorem}
  \label{thm:chernoff}
  Let $X_1, \ldots, X_n$ be negatively dependent random variables such that $X_i \in [0,1]$ with probability 1. Define $X = \sum_{i=1}^n X_i$ and let $\mu = \mathbb{E} X$.
  Then, for any $\epsilon > 0$, we have
  \[
      \Pr[X \geq (1 + \epsilon) \mu ] \leq \exp\left( - \frac{\min\{\epsilon, \epsilon^2\} \mu}{3} \right).
  \]
\end{theorem}

\end{document}